\documentclass[aps,floatfix,rmp,reprint,longbibliography,showkeys]{revtex4-1}

\usepackage{amssymb,amsmath,amsthm}
\usepackage{graphicx,color}

\usepackage{xspace}

\newcommand{\ie}{i.e.\ }
\newcommand{\eg}{e.g.,\ }

\newcommand{\ITEMGAP}{\vspace{6pt}}

\newcommand{\NEWAXIOM}[3]{%
  \ITEMGAP
  \begin{description}
  \label{#2}%
  \item[\underline{#1}]{#3}
  \ITEMGAP
  \end{description}
}

\newcommand{\MATH}[1]{\ensuremath{#1}\xspace}
\newcommand{\Tuple}[1]{\MATH{\left\langle{#1}\right\rangle}}
\newcommand{\EMPH}[1]{\textbf{#1}}

\newcommand{\Abs}[1]{\MATH{\left|{#1}\right|}}

\newcommand{\MATHBB}[1]{\MATH{\mathbb{#1}}}

\newcommand{\Rset}{\MATHBB{R}}

\newcommand{\MATHIT}[1]{\MATH{\mathit{#1}}}
\newcommand{\speed}{\MATHIT{v}}
\newcommand{\m}{\MATHIT{m}}

\newcommand{\MATHSF}[1]{\MATH{\mathsf{#1}}}

\newcommand{\inecoll}[4]{\MATHSF{inecoll}_{#1}(#2#3\mathbin{:}#4)}

\newcommand{\Time}{\MATHSF{time}}
\newcommand{\Space}{\MATHSF{space}}

\newcommand{\MATHBF}[1]{\MATH{\mathbf{#1}}}

\newcommand{\vel}{\MATHBF{v}}
\newcommand{\bv}{\MATHBF{v}}
\newcommand{\bp}{\MATHBF{p}}

\newcommand{\BAR}[1]{\MATH{\bar{#1}}}

\newcommand{\vx}{\BAR{x}}
\newcommand{\vy}{\BAR{y}}
\newcommand{\vz}{\BAR{z}}

\newcommand{\leteq}{\MATH{\mathrel{:=}}}

\newcommand{\Ip}{\MATHSF{Ip}}

\newcommand{\IOb}{\MATHSF{IOb}} 

\newcommand{\Q}{\MATHIT{Q}} 
\newcommand{\ax}[1]{\textcolor{axcolor}{\MATHSF{#1}}}

\newcommand{\wl}{\MATH{\mathit{w}\ell}}

\definecolor{firebrick}{rgb}{0.7,0.13,0.13}

\definecolor{thmcolor}{rgb}{0,0,.4} 
\definecolor{remarkcolor}{rgb}{0,.2,0} 
\definecolor{proofcolor}{rgb}{.4,0,0} 
\definecolor{quecolor}{rgb}{.2,.2,0} 
\definecolor{axcolor}{rgb}{.3,0,.3}
\definecolor{thmbgcolor}{rgb}{0.9,0.9,1} 
\definecolor{rmbgcolor}{rgb}{0.9,1,0.9} 
\definecolor{proofbgcolor}{rgb}{1,0.9,0.9}

\definecolor{qcolor}{rgb}{0,0.4,0}
\definecolor{lqcolor}{rgb}{0,0.6,0}
\definecolor{phcolor}{rgb}{.6,0,0}
\definecolor{lphcolor}{rgb}{.7,0,0}
\definecolor{evcolor}{rgb}{.5,.3,.2}
\definecolor{obcolor}{rgb}{0,0.4,.5}
\definecolor{lobcolor}{rgb}{0,0.6,.75}
\definecolor{iobcolor}{rgb}{0,0,.6}
\definecolor{liobcolor}{rgb}{0,0,.7}
\definecolor{axbgcolor}{rgb}{1,.7,1}

\theoremstyle{definition}

\newtheorem{thm}{\textcolor{thmcolor}{Theorem}}[section]

\newtheorem{prop}[thm]{\textcolor{thmcolor}{Proposition}}

\theoremstyle{remark}

\begin{document}

\title{On the Possibility and Consequences of Negative Mass}

\author{J. X. \surname{Madar\'asz}}
\email{madarasz.judit@renyi.mta.hu}
\author{G. \surname{Sz\'ekely}}
\email{szekely.gergely@renyi.mta.hu}
\affiliation{Alfr\'ed R\'enyi Institute of Mathematics of the Hungarian Academy of Sciences, P.O.Box 127, Budapest 1364, Hungary}
\author{M. \surname{Stannett}}
\email{m.stannett@sheffield.ac.uk}
\affiliation{Department of Computer Science, University of Sheffield, Sheffield, S1 4DP, UK}

\date{\today{}}

\begin{abstract}
We investigate the possibility and consequences of the existence of particles having negative relativistic masses, and show that their existence implies the existence of faster-than-light particles (tachyons). Our proof requires only two postulates concerning such particles: that it is possible for particles of any (positive, negative or zero) relativistic mass to collide inelastically with `normal' (\ie positive relativistic mass) particles, and that four-momentum is conserved in such collisions. 
\end{abstract}



\maketitle

\tableofcontents{}

\section{Introduction}

In his well-known relativistic analysis of negative-mass particles, \textcite[p.~428]{Bondi57} successfully constructed a ``world-wide nonsingular solution of Einstein's equations containing two oppositely accelerated pairs of bodies, each pair consisting of two bodies of opposite sign of mass''. Jammer has discussed the historical and philosophical context of negative mass at length. While stressing the fact that no negative-mass particle has yet been observed experimentally, he notes \cite[pp.~129--130]{Jam00} that ``no known physical law precludes the existence of negative masses''. More recently, Bellet\^ete and Paranjape have demonstrated in a general relativistic setting that Schwarzchild solutions exist representing matter distributions which are ``perfectly physical'', despite describing a negative mass Schwarzschild-de Sitter geometry outside the matter distribution \cite{BP13}. 

In this paper, we investigate further whether it is possible for a particle to have negative relativistic mass by considering the formal logical consequences of the existence of such particles. We show formally that if such particles exist, and provided they can collide inelastically (\ie fuse together) with `normal' particles in collisions that conserve four-momentum, then faster-than-light (FTL) particles must also exist. We prove our statement by showing how, given any negative mass particle $a$ with known 4-momentum, it is possible to specify a suitable positive mass particle $b$, such that the inelastic collision of $a$ with $b$ would generate an FTL body.

Our approach is firmly based in an axiomatic logical framework for relativity theory, thereby avoiding the use of unstated and potentially unjustifiable assumptions in deriving our results. Avoiding such assumptions, and in particular the blanket assumption that negative-mass particles cannot exist, is important in this context, since it allows us to provide explanations as to why such phenomena may or may not be physically feasible. In contrast, if we simply assert from the outset that negative mass is unphysical, the only answer we can give to the question ``why?'', is ``because we say so''. For example, it might be argued informally that the entailed existence of FTL particles, proven in this paper, would itself entail the possibility of causality paradoxes, so that the consequences of negative mass particles are not `reasonable'. But informal arguments of this nature can be flawed: using our formal approach, we and our colleagues have recently shown \cite{CQGPaper} that spacetime (of any dimension $n+1$) can be equipped with particles and observers in such a way that faster-than-light motion is possible, but this does \textit{not} lead to the `time travel' situations that give rise to causality problems. Consequently, the fact that negative-mass particles entail the existence of FTL particles cannot, of itself, be used to argue against their existence.

Formal axiomatization also allows us to address consistency issues and what-if scenarios. It is possible to show, for example, that the consistency of relativistic dynamics with interacting particles having negative relativistic masses follows by a straightforward generalization of the  model construction used by \textcite{FTLconsSRDyn} to prove the consistency of relativistic dynamics and interacting FTL particles. The same approach allows us to derive and prove the validity of key relativistic formulae. For example, all inertial observers of any particle must agree on the value of $\m \cdot \sqrt{\Abs{ 1-v^2 }}$, where $\m$ is the particle's relativistic mass and $v$ its speed ($c = 1$), confirming the widely-held belief that both the observed relativistic mass and the observed relativistic momentum of a positive-mass FTL particle must decrease as its relative speed increases \cite{MSS14}. 

We introduce our result in two stages. In Section~\ref{sec: 3 ways}, we show informally that there are several simple ways to create FTL particles using inelastic collisions between positive and negative relativistic mass particles. Then in Section~\ref{sec: formal proof}, we reconstruct our informal arguments within an axiomatic framework so as to provide a complete formal proof of our central claim, that the existence of particles with negative relativistic mass necessarily entails the existence of FTL particles.

Restating and proving our statements formally in this way has a further advantage over the informal approach. The mechanics of proof construction require us to identify all of the tacit assumptions underpinning our informal arguments, thereby revealing which assumptions are relevant and which are unwarranted or unnecessary. This is, fortunately, a task in which automated interactive theorem provers (\eg Isabelle/HOL \cite{Isabelle}) are increasingly able to assist, both in terms of proof production and automatic checking of correctness. This approach is already leading to the production and machine-verification of non-trivial relativistic theorems \cite{SN13}.

\section{Generating FTL particles from negative mass particles} 
\label{sec: 3 ways}

Let us assume that particles do indeed exist with negative relativistic mass, and that it is possible for such particles to collide inelastically with `normal' particles. As we now illustrate informally, the existence of FTL particles (tachyons) follows almost immediately, provided we assume that four-momentum is conserved in such collisions. For simplicity, we take $c = 1$.

Recall first that the \textit{four momentum} of a particle $b$ is the four-dimensional vector $(\m,\bp)$, where \m is its relativistic mass and \bp its linear momentum (as measured by some inertial observer whose identity need not concern us). Notice also that the particle $b$ is a tachyon if and only if $|\m|<|\bp|$ (\ie its observed speed is greater than $c = 1$), and that all inertial observers agree as to this judgment (if one inertial observer considers $b$ to be travelling faster than light, they all do --- this is because all inertial observers consider each other to be travelling slower than light relative to one another. For a machine-verified proof of this assertion, see \cite{SN13}).

In each of the following thought experiments, we will always understand `mass' to mean `relativistic mass'. We will assume the existence of two colliding particles $a$ and $b$, where $a$ has negative mass $\m < 0$ and $b$ has positive mass $M > 0$, which move along the same spatial line (though possibly in opposite directions). Taking the common line of travel to be the $x$-axis, positive in the direction of $b$'s travel, the four-momenta of $a$ and $b$ can be written $(\m,p,0,0)$ and $(M,P,0,0)$, respectively, for suitable values of $p$ and $P$. Assuming that four-momentum is conserved during the collision, the four-momentum of the particle $c$ generated by the fusion of $a$ and $b$ will be $(M+\m,P+p,0,0)$, and this particle will be a tachyon provided 
\begin{equation}
  |M+\m| < |P+p| \label{eqn:thoughts}
\end{equation}
If this tachyon has negative mass and positive momentum, it moves in the negative $x$-direction (it is an unusual property of negative-mass particles that their velocity and momentum vectors point in opposite directions); if it has positive-mass and positive momentum it moves in the positive $x$-direction. By definition, $M > 0 > m$, and $b$ has both positive mass ($M > 0$) and positive momentum ($P > 0$), since its motion defines the positive $x$-direction.

\subsection{First thought experiment}
\label{subsec:thexp1}

Suppose $a$ travels slower than light, while $b$ moves at light-speed, so that the four-momenta of $a$ and $b$ can be written $(m,p,0,0)$ and $(M,M,0,0)$, respectively. According to (\ref{eqn:thoughts}), the particle created by their collision will be a tachyon provided 
\begin{equation}
  |M+m| < |M+p| \label{eqn:thought1}
\end{equation}
There are various ways in which this can happen, depending on the values of \m and $p$ (see Fig.~\ref{fig-thexp1}). Notice that $|p| < |\m|$ since $a$ travels slower than light.

\begin{figure}[htp]
\includegraphics*[width=\linewidth]{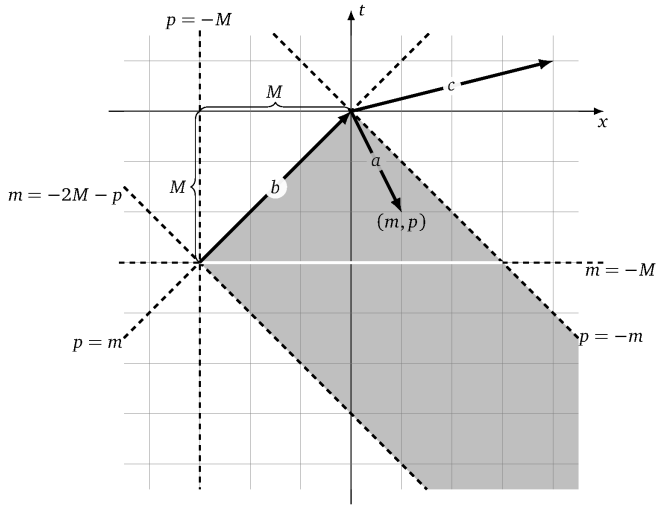}
\caption{\label{fig-thexp1} Illustration for generating an FTL particle by colliding a negative relativistic mass particle with a particle moving with the speed of light.}
\end{figure}

\begin{figure}[htp]
\includegraphics*[width=\linewidth]{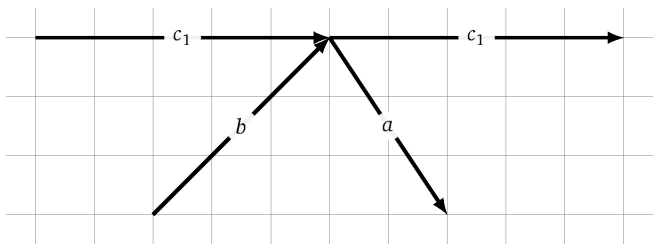}
\caption{\label{fig-ambiguous} The ``inelastic'' collision of two particles having opposite relativistic masses is ambiguous in the sense that in this case we have two possible outcomes satisfying the conservation of four-momentum.}
\end{figure}

\textbf{Case 1.}
If $|\m| > M$, then $|\m+M| = -\m - M$. According to (\ref{eqn:thought1}), a tachyon will be generated provided $-\m - M < |M+p|$. Since $\m+M < 0$ in this case, any resulting tachyon will have negative relativistic mass.

\begin{itemize}
\item
  If $p > -M$, so that $|M+p| = M+p$, then a tachyon forms provided $-\m-M < M+p$, \ie $M > -(\m + p)/2$. Since $M + p > 0$, the resulting tachyon will have positive linear momentum and will consequently move in the negative $x$-direction. The case $p = 0$ corresponds to the collision of a light-speed particle with a stationary negative-mass particle, in which case a tachyon will form provided $M > -\m/2$. See Fig.~\ref{fig-thexp1}.
\item
  If $p \leq -M$, so that $|M+p| = -p-M$, then a tachyon forms provided $-\m-M < -p-M$, \ie $-\m < -p$. This situation is impossible, since $\m$ and $p$ are both non-positive. This means that the requirement can be written $|\m| < |p|$, but we already know that $|p| < |\m|$, since $a$ travels slower than light.
\end{itemize}
  
\textbf{Case 2.}
   If $|\m| < M$, a tachyon is always formed, because:
\begin{itemize}
\item
   if $p \geq 0$, then $|M + p| = M + p$ and (\ref{eqn:thought1}) becomes $M+\m < M + p$, which follows directly from $\m < 0 \leq p$; 
\item
   if $p < 0$, then $|M+p| = M-|p|$ (this is positive because $|p| < |m| < M$) and (\ref{eqn:thought1}) becomes $M+\m < M - |p|$, which follows directly from $-m = |\m| > |p|$, \ie $m < -|p|$.
\end{itemize}
The resulting tachyon $c$ will have positive relativistic mass (because $M > -\m$) and positive linear momentum (because $(M+p) \geq \left(M - |p|\right) > \left(M - |\m|\right) > 0$). It will therefore move in the positive $x$-direction.

\textbf{Case 3.}
The case when $|\m| = M$, \ie $M = -\m$, is ambiguous irrespective of the velocities of the colliding particles $a$ and $b$. Since $M+\m = 0$ and $|p| < |\m| = M$, the linear momentum $M + p$ of particle $c$ must be positive, even though it has zero relativistic mass. In terms of the space-time diagram (Fig.~\ref{fig-ambiguous}), this means that the resulting particle's worldline is horizontal, \ie it `moves' with infinite speed. In these circumstances, the question whether $c$ moves in the positive or negative $x$-direction is meaningless. However, like other observer-dependent concepts such as simultaneity or the temporal ordering of events, this indeterminacy does not lead to a logical contradiction \cite{FTLconsSRDyn}.

\subsection{Second thought experiment}
\label{subsec:thexp2}

Suppose $b$ is stationary, \ie $P = 0$. By arguments similar to those above, this will result in an FTL particle $c$ whenever $|m|+|\bp|>M>|m|-|\bp|$, and its direction of travel will be determinate provided $M\neq -m$. See Fig.~\ref{fig-thexp2}.
     
\begin{figure}[htp]
\includegraphics*[width=\linewidth]{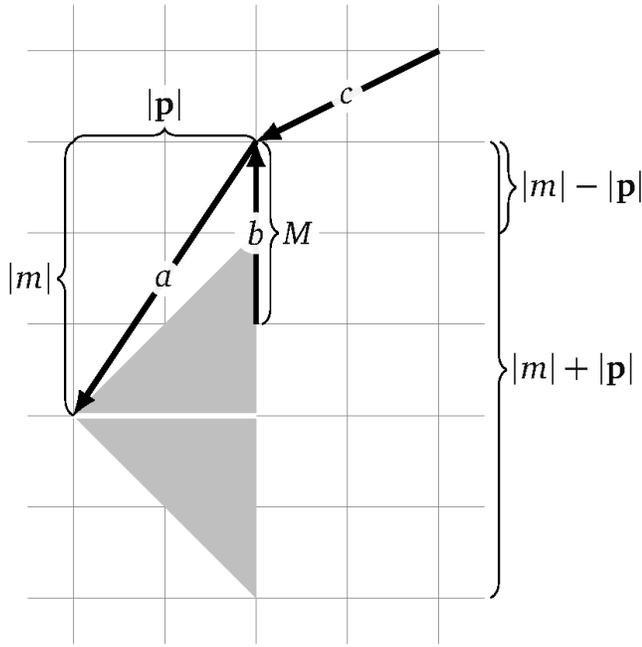}
\caption{\label{fig-thexp2} Illustration for generating an FTL particle by colliding a negative relativistic mass particle with a stationary particle of positive relativistic mass.}
\end{figure}

\subsection{Third thought experiment}
\label{subsec:thexp3}

Suppose $a$ and $b$ have similar, but oppositely-signed, masses, and that they collide `head-on' while travelling with equal speeds in opposite directions (relative to some observer, whose identity need not concern us). If the difference in the absolute values of their masses is small relative to their common speed, the resulting particle will be FTL.

More precisely, suppose $a$ and $b$ have equal speed $v > 0$ when they collide, and that $(1+\epsilon)|\m| < M < (1+2\epsilon)|m|$ for some $\epsilon$ in the range $0 < \epsilon < v$ (\ie $|\m|$ and $M$ are roughly equal). Since $|m| = -m$, it follows that $\epsilon|\m| <  M+m < 2\epsilon|\m|$, and hence that $|M+m| < 2\epsilon|\m|$. Equation (\ref{eqn:thoughts}) now follows, since $|P+p| = |m|v + Mv > (2+\epsilon)|m|v > 2|m|v$, and hence
\begin{equation}
  |P+p| > \left(\frac{|M+m|}{2\epsilon|m|}\right)(2|m|v)
        > |M+m| .
\end{equation}

\subsection{FTL particle creations requiring negative relativistic mass}

We have seen above that the existence of negative-mass particles implies the existence of FTL particles. Conversely, it is easy to see that an inelastic collision between two slower-than-light particles having positive relativistic masses always leads to a slower-than-light particle. Consequently, the only way in which an inelastic collision between slower-than-light particles can create an FTL particle is if incoming particles can have negative relativistic masses. 

In particular, if we impose the condition that such collisions are the only mechanism by which FTL particles can be created, then the existence of FTL particles implies the existence of negative-mass particles. While this suggests that tachyons and negative-mass particles are equally `exotic', this is, of course, not the case, since the argument that FTL particles require the existence of negative-mass particles relies on the assumption that inelastic collisions are the only mechanism by which FTL particles can be created. 

This is no means a trivial assumption; indeed we have demonstrated elsewhere a consistent model of spacetime in which FTL particles exist, but in which no collisions are posited \cite{CQGPaper}.

\section{Axiomatic reconstruction}
\label{sec: formal proof}

An important benefit of reconstructing the previous ideas within an axiomatic framework is that we can localize the required basic assumptions more precisely. Therefore, in this section, we make the intuitive ideas and arguments of the above sections precise in an axiomatic framework. Readers interested in the wider context are referred to \cite{FTLconsSRDyn,SN13,MSS14}.

\subsection{Quantities and Vector Spaces}
To formulate the intuitive image above, we need some structure of numbers describing physical quantities such as coordinates, relativistic masses and momenta. Traditional accounts of relativistic dynamics take for granted that the basic number system to be used for expressing measurements (lengths, masses, speeds, etc.) is the field \Rset of real numbers, but this assumption is far more restrictive than necessary. Instead, we will only assume that the number system is a linearly ordered field $Q$ equipped with the usual constants, zero ($0$) and one ($1$); the usual field operations, addition ($+$), multiplication ($\cdot$) and their inverses; and the usual ordering ($\le$) and its inverse; we also assume that the field is \textit{Euclidean}, \ie positive quantities have square roots. Formally, this is declared as an axiom:

\NEWAXIOM{\ax{AxEField}}{AxEField}
{
  The structure  \Tuple{\Q, 0, 1, +, \cdot, \le}  of quantities is a linearly ordered field (in the algebraic sense) in which all non-negative numbers have square roots, \ie $(\forall x \in \Q)((0 \le x) \Rightarrow (\exists y \in \Q)(x=y^2))$.  
}

We write $\sqrt{x}$ for this root, which can be assumed without loss of generality to be both unique and non-negative (regarding machine-verified proofs of this and other relevant claims concerning Euclidean fields, see \cite{SN13}).

\subsection{Inertial particles and observers}

We denote the set of physical \textit{bodies} (things that can move) by $B$. This includes the sets $\IOb \subseteq B$ of \EMPH{inertial observers}, $\Ip \subseteq B$ of \EMPH{inertial particles}. Given any inertial observer $k \in \IOb$ and inertial particle $b \in \Ip$, we write $\wl_k(b) \subseteq \Q^4$ for the \EMPH{worldline} of particle $b$ as observed by $k$. The coordinates of $\vx\in\Q^n$ are denoted by $x_1$, $x_2$, \ldots, $x_n$.

The following axiom  asserts that the motion of inertial particles are uniform and rectilinear according to inertial observers. 

\NEWAXIOM{\ax{AxIp}}{AxIp}
{
  For all  $k\in\IOb$ and $b\in\Ip$, the worldline $\wl_k(b)$ is either a line, a half-line or a line segment\footnote{%
  Taking $\vx$ and $\vy$ to be of sort $\Q^4$, and $\lambda$ to be of sort \Q, these concepts are defined formally as follows. A \textit{line} is a set of the form $\{ \vz ~|~ (\exists \vx,\vy,\lambda) ( \vz=\lambda \vx+(1-\lambda)\vy ) \}$. A \textit{half-line} is a set of the form $\{ \vz ~|~ (\exists \vx,\vy,\lambda) ((0\le \lambda) \& (\vz=\lambda \vx+(1-\lambda)\vy) )\}$. A \textit{line segment} is a set of the form  $\{ \vz | (\exists \vx,\vy,\lambda) ((0\le \lambda \le 1) \& (\vz=\lambda \vx+(1-\lambda)\vy) ) \}$. }.
}

Suppose observer $k \in \IOb$ sees particle $b \in \Ip$ at the distinct locations $\vx, \vy \in \Q^4$. Then its \EMPH{velocity} according to $k$ is the associated change in spatial component divided by the change in time component,
\[
   \vel_k(b) \leteq \begin{cases}
      \frac{\Space(\vx,\vy)}{\Time(\vx,\vy)} & \text{ if $\Time(\vx,\vy)\neq 0$ } \\
      \mathsf{undefined} & \text{ otherwise }
      \end{cases}
\]
where $\Space(\vx,\vy)\leteq (x_2-y_2,x_3-y_3,x_4-y_4)$ and $\Time(\vx,\vy)\leteq x_1-y_1$. The length\footnote{The \EMPH{Euclidean length}, $|\vx|$, of a vector $\vx$ is the non-negative quantity $|\vx| = \sqrt{ x_1^2 + \dots + x_n^2  }$.} of the velocity vector (if it is defined) is the particle's \EMPH{speed},
\[
 \speed_k(b) \leteq |\vel_k(b)| .
\]
 By \ax{AxIp}, these concepts are well-defined because $\wl_k(b)$ lies in a straight line, so that all choices of distinct \vx and \vy give the same results.

If $\vel_k(b)$ is defined, we say that $b$ is observed by $k$ to have \EMPH{finite speed}, and write $\speed_k(b) < \infty$. The anomalous case $\Time(\vx,\vy) = 0$ corresponds to a situation where all points in $\wl_k(b)$ are simultaneous from $k$'s point of view, so that $k$ considers the particle to require no time at all to travel from one spatial location to another.

\subsection{Collision axioms}

In this subsection, we introduce some very simple axioms concerning the dynamics of collisions, and show that the existence of negative relativistic mass implies the existence of faster-than-light (FTL) inertial particles.

Suppose an inertial observer $k$ sees two inertial bodies traveling at finite speed fuse to form a third one at some point \vx. In this case, the worldlines of the two incoming particles terminate at \vx, while that of the outgoing particle originates there. Formally, we say that an inertial particle $b$ is \EMPH{incoming} at \vx (according to $k$) provided $\vx \in \wl_k(b)$ and \vx occurs strictly later than any other point on $\wl_k(b)$, \ie $\vy \in \wl_k(b) \enskip\&\enskip \vy \neq \vx \;\Rightarrow\; y_1 < x_1$. \EMPH{Outgoing} bodies are defined analogously. An \EMPH{inelastic collision} between two inertial particles $a$ and $b$ (according to observer $k$) is then a scenario in which there is a unique additional particle $c \in \Ip$ and a point \vx such that $a$ and $b$ are incoming at $\vx$, $c$ is outgoing at $\vx$. We write  $\inecoll kabc$ to denote that the distinct inertial particles $a$ and $b$ \EMPH{collide inelastically}, thereby generating inertial particle $c$ (according to observer $k$). The \EMPH{relativistic mass} of inertial particle $b$ according to observer $k$ is denoted by $\m_k(b)$.

\NEWAXIOM{\ax{ConsFourMomentum}}{ConsFourMomentum}{Four-momentum is conserved in inelastic collisions of inertial particles according to inertial observers, \ie
\[\begin{aligned}
&\quad \inecoll kabc \Rightarrow \\
&\quad\quad \m_k(c)=\m_k(a)+\m_k(b)  \quad \&\\
&\quad\quad\quad \m_k(c)\vel_k(c)=\m_k(a)\vel_k(a)+\m_k(b)\vel_k(b) 
\end{aligned}\]
}

The next axiom, \ax{AxInecoll}, states that inertial particles moving with finite speeds can be made to collide inelastically in any frame in which their relativistic masses are not equal-but-opposite. Since a collision of particles having equal but opposite relativistic masses does not lead to an inelastic collision according to our formal definition, we do not include this case in this axiom (this does not mean that such particles cannot collide, just that such a collision will not comply with our definition of inelasticity in the associated frame because the third participating particle has infinite speed).

\begin{figure}[htp]
\includegraphics*[width=\linewidth]{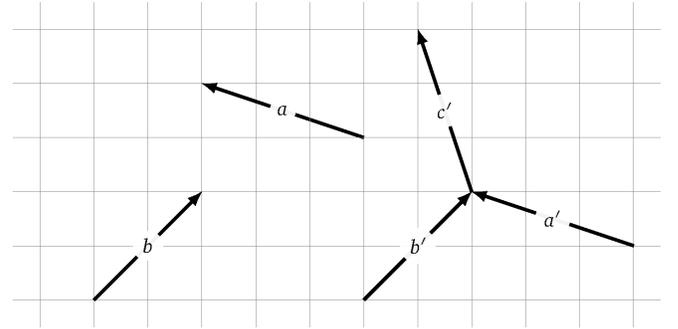}
\caption{\label{fig-axinecoll} Illustration for axiom \ax{AxInecoll}}
\end{figure}

\NEWAXIOM{\ax{AxInecoll}}{axinecoll}
{
If $k\in\IOb$ and  $a,b\in\Ip$ such that $\speed_k(a)<\infty$, $\speed_k(b)<\infty$ and $\m_k(a) + \m_k(b)\neq 0$, then there are $a',b'\in\Ip$ such that $a'$ and $b'$ collide inelastically, with $\m_k(a')=\m_k(a)$, $\vel_k(a')=\vel_k(a)$, $\m_k(b')=\m_k(b)$ and $\vel_k(b')=\vel_k(b)$.\footnote{Because here we use the framework of \cite{Dyn-studia}, we cannot distinguish actual and potential particles only existing in some thought experiments.  See \cite{MoSz13} for an axiomatic framework where this distinction can be made and the idea of thought experiments can be captured formally.} See Fig.~\ref{fig-axinecoll}.
}

\subsection{Formulating the thought experiments}

Here we are going to formalize and prove the thought experiments of Subsections~\ref{subsec:thexp1}, \ref{subsec:thexp2} and \ref{subsec:thexp3}.

Formula \ax{\exists NegMass} below says that there is at least one inertial particle of finite speed and negative relativistic mass.

\NEWAXIOM{\ax{\exists NegMass}}{NegMass}{There are $k\in\IOb$ and $b\in\Ip$ such that $\m_k(a)<0$ and $\speed_k(a)<\infty$.}

Formula \ax{\exists FTL Ip} below says that there is at least one faster than
light inertial particle.
\NEWAXIOM{\ax{\exists FTL Ip}}{FTL Ip}{There are  $k\in\IOb$ and  $b\in\Ip$ such that $1<\speed_k(b)<\infty$.}

Axiom \ax{AxThExp_1} below says that the thought experiment described in Subsection~\ref{subsec:thexp1} can be done by asserting that inertial observers can send out particles moving with the speed of light $1$ in any direction every where having arbitrary positive relativistic mass. 
 \NEWAXIOM{\ax{AxThExp_1}}{ThExp1}{For $k\in\IOb$, $m\in\Q$ and $\bv\in\Q^3$ for which $m>0$ and  $|\bv|=1$, there is $b\in\Ip$ such that $\vel_k(b)=\bv$ and $\m_k(b)=m$.}

\begin{prop}
\label{prop1}
Assume \ax{ConsFourMomentum}, \ax{AxEField}, \ax{AxIp}, 
\ax{AxInecoll}, \ax{AxThExp_1}. Then
\begin{equation}
\label{impeq2}
\ax{\exists NegMass}\ \Rightarrow\ \ax{\exists FTL\Ip}.
\end{equation}
\end{prop}

\begin{proof}
By axiom \ax{\exists NegMass}, there is an inertial observer $k$ and inertial particle $a$ such that  $\m_k(a)<0$ and $\speed_k(a)<\infty$. Let $\bv\in\Q^3$ for which $|\bv|=1$. Then by axiom \ax{AxThExp_1}, there is an inertial particle $b$ such that $\m_k(b)=-2\m_k(a)$ and 
\begin{equation*}
\vel_k(b)=
\begin{cases}
\bv & \text{if } \speed_k(a)=0,\\
\frac{-\vel_k(a)}{\speed_k(a)} & \text{if } \speed_k(a)\neq 0.
\end{cases}
\end{equation*}
By axiom \ax{AxInecoll}, there are inelastically colliding inertial particles $a'$, $b'$ and $c'$ such that $\inecoll k{a'}{b'}{c'}$, $\m_k(a')=\m_k(a)$, $\vel_k(a')=\vel_k(a)$, $\m_k(b')=\m_k(b)$ and $\vel_k(b')=\vel_k(b)$. By \ax{ConsFourMomentum},
\begin{equation}\begin{aligned}
\m_k(c')&=\m_k(a') + \m_k(b')\\
        &=\m_k(a)+\m_k(b)  
         = -\m_k(a)
\end{aligned}\end{equation}
and   
\begin{equation}
\m_k(c')\vel_k(c') = \begin{cases}
- 2\m_k(a)\bv & \text{if $\speed_k(a)=0$},\\
\m_k(a)\vel_k(a)+2\m_k(a)\frac{\vel_k(a)}{\speed_k(a)} & \text{if $\speed_k(a)\neq 0$}.
\end{cases}
\end{equation}
Hence 
\begin{equation}
\vel_k(c')=
\begin{cases}
2\bv  & \text{if } \speed_k(a)=0,\\
-(\speed_k(a)+2) \frac{\vel_k(a)}{\speed_k(a)}  & \text{if } \speed_k(a)\neq 0.
\end{cases}
\end{equation}
Therefore, $\speed_k(c')=|\vel_k(c')|>1$ and $\speed_k(c')<\infty$; and this is what we wanted to prove. 
\end{proof}

Axiom \ax{AxThExp_2} below ensures the existence of the particle having positive relativistic mass used in the thought experiment described in Subsection~\ref{subsec:thexp2}. 

\NEWAXIOM{\ax{AxThExp_2}}{ThExp2}{For every  $k\in\IOb$ and $m>0$, there is $b\in\Ip$ such that $\speed_k(b)=0$ and $\m_k(b)=m$.}

Formula \ax{\exists MovNegMass} below asserts that there is at least one moving inertial particle of finite speed and negative relativistic mass. 
\NEWAXIOM{\ax{\exists MovNegMass}}{MovNegMass}{There are $k\in\IOb$ and $b\in\Ip$ such that $\m_k(b)<0$ and $0<\speed_k(b)<\infty$.}
For the sake of economy, we use axiom \ax{\exists MovNegMass} instead of \ax{\exists NegMass} because in this case we do not have to assume anything about the possible motions of inertial observers or the transformations between their worldviews. We note, however, that these two axioms are clearly equivalent in both Newtonian and relativistic kinematics (assuming that inertial observers can move with respect to each other).

\begin{prop}
\label{prop2}
Assume \ax{ConsFourMomentum}, \ax{AxEField}, \ax{AxIp},
 \ax{AxInecoll}, \ax{AxThExp_2}. Then
\begin{equation}
\label{impeq1}
\ax{\exists MovNegMass}\ \Rightarrow\ \ax{\exists FTL\Ip}.
\end{equation}
\end{prop}

\begin{proof}
By axiom \ax{\exists MovNegMass}, there is an inertial observer $k$ and inertial particle $a$ such that  $\m_k(a)<0$ and $0<\speed_k(a)<\infty$.  By axiom \ax{AxThExp_2}, there is an inertial particle $b$ such that $\m_k(b)=-\m_k(a)\left(1+\speed_k(a)/2\right)$ and $\speed_k(b)=0$. By axiom \ax{AxInecoll}, there are inelastically colliding inertial particles $a'$, $b'$ and $c'$ such that $\inecoll k{a'}{b'}{c'}$, $\m_k(a')=\m_k(a)$, $\vel_k(a')=\vel_k(a)$, $\m_k(b')=\m_k(b)$ and $\vel_k(b')=\vel_k(b)$. By \ax{ConsFourMomentum}, 
\begin{equation}\begin{aligned}
\m_k(c') &=\m_k(a')+ \m_k(b') \\
         &=\m_k(a)+ \m_k(b) 
          =\frac{-\m_k(a)\speed_k(a)}{2} 
\end{aligned}\end{equation}
and     
\begin{equation}
\m_k(c')\vel_k(c')=\m_k(a)\vel_k(a).
\end{equation}
It follows that
\[
   \vel_k(c')=-2\frac{\vel_k(a)}{\speed_k(a)},
\] 
and hence that $\speed_k(c') = 2 > 1$, which is what we wanted to prove. 
\end{proof}  

Finally let us introduce the following axiom ensuring the existence of the particles having positive relativistic mass needed in the thought experiment of Subsection~\ref{subsec:thexp3}.
\NEWAXIOM{\ax{AxThExp_3}}{ThExp3}{For all $\varepsilon>0$, $k\in\IOb$ and $a\in\Ip$, there is $b\in\Ip$ such that $(1+\varepsilon)|\m_k(a)|<\m_k(b)<(1+2\varepsilon)|\m_k(a)|$ and $\vel_k(a)=-\vel_k(b)$.}

\begin{prop}
\label{prop3}
Assume \ax{ConsFourMomentum}, \ax{AxEField}, \ax{AxIp}, 
\ax{AxInecoll}, \ax{AxThExp_3}. Then
\begin{equation}
\label{impeq3}
\ax{\exists MovNegMass}\ \Rightarrow\ \ax{\exists FTL\Ip}.
\end{equation}
\end{prop}

\begin{proof}
By axiom \ax{\exists MovNegMass}, there is an inertial observer $k$ and inertial particle $a$ such that  $\m_k(a)<0$ and $0<\speed_k(a)<\infty$. Let $0<\varepsilon<\speed_k(a)$.  Then by axiom \ax{AxThExp_3}, there is an inertial particle $b$ such that $(1+\varepsilon)|\m_k(a)|<\m_k(b)<(1+2\varepsilon)|\m_k(a)|$ and $\vel_k(b)=-\vel_k(a)$.  

\begin{figure}[htp]
\includegraphics*[width=\linewidth]{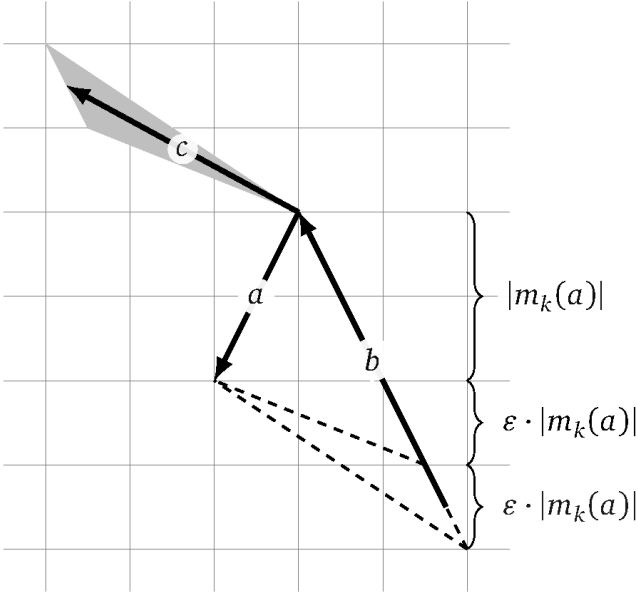}
\caption{\label{fig-thexp3} Illustration for the proof of Proposition~\ref{prop3} }
\end{figure}

By axiom \ax{AxInecoll}, there are inelastically colliding inertial particles $a'$, $b'$ and $c'$ such that $\inecoll k{a'}{b'}{c'}$, $\m_k(a')=\m_k(a)$, $\vel_k(a')=\vel_k(a)$, $\m_k(b')=\m_k(b)$ and $\vel_k(b')=\vel_k(b)$. By \ax{ConsFourMomentum}, 
\begin{equation}
\varepsilon |\m_k(a)|  <  |\m_k(c')|  <  2\varepsilon |\m_k(a)|
\end{equation}
and     
\begin{equation}
2|\m_k(a)|\speed_k(a) < (2+\varepsilon)|\m_k(a)|\speed_k(a)  <  |\m_k(c')\vel_k(c')|.
\end{equation}
Hence 
\begin{equation}
\speed_k(c')=|\vel_k(c')|>\frac{2|\m_k(a)|\speed_k(a)}{2\varepsilon |\m_k(a)|}>\frac{\speed_k(a)}{\varepsilon}.
\end{equation}
Therefore, $1<\speed_k(c')<\infty$; and this is what we wanted to prove. 
\end{proof}

\section{Concluding remarks}

Using only basic postulates concerning the conservation of four-momentum, we have axiomatically shown that the existence of particles having negative relativistic masses implies the existence of FTL particles. The following are the two most straightforward applications of this result.
\begin{itemize}
\item If experiment eventually shows the existence of particles having negative masses, then we will know that FTL particles must also exist. If evidence exists suggesting otherwise, our approach would then imply that one or more of the natural assumptions encoded in our axioms must to be false. This in turn would provide information suitable for guiding further experimentation.
\item Similarly, if we can prove that FTL particles cannot exist, and \emph{no} evidence can be found suggesting that the natural physical assumptions encoded by our axioms are invalid, then this can be used to prove the non-existence of particles having negative masses.
\end{itemize}

It is also worth noting that our axioms are so general that they are compatible with both Newtonian and relativistic kinematics, so that our method can be used to derive predictions for both settings. Moreover, we have made no restrictions on the worldview transformations between inertial observers. A benefit of being so parsimonious with the basic assumptions is that it makes results obtained using our axiomatic method that much more difficult to challenge, because so few basic assumptions have been made concerning physical behaviours in the ``real world''.

\section*{Acknowledgments}
This research was partially supported under the Royal Society International Exchanges Scheme (ref. IE110369) and by the Hungarian Scientific Research Fund for basic research grants No. T81188 and No. PD84093, as well as by a Bolyai grant for J. X. Madar\'asz.

\bibliography{LogRelBib}

\end{document}